\documentclass[11pt,a4paper]{amsart}

\usepackage[utf8]{inputenc}
\usepackage[T1]{fontenc}
\usepackage[english]{babel}
\usepackage{amsmath}
\usepackage{amsfonts}
\usepackage{ae}
\usepackage{units}
\usepackage{icomma}
\usepackage{color}
\usepackage{graphicx}
\usepackage{bbm}
\usepackage{amsthm}
\usepackage{amssymb}
\usepackage{cancel}
\usepackage{fullpage}
\usepackage{upgreek}
\usepackage{algorithm}
\usepackage{algorithmic}
\usepackage{type1cm}
\usepackage{eso-pic}
\usepackage[breaklinks,pdfpagelabels=false]{hyperref}

\newcommand{\bm}[1]{\mbox{\boldmath $ {#1} $}}

\renewcommand{\epsilon}{\varepsilon}

\newtheorem{theorem}{Theorem}[section]
\newtheorem{lemma}[theorem]{Lemma}

\theoremstyle{definition}
\newtheorem{rem}[theorem]{Remark}
\newtheorem{defi}[theorem]{Definition}
\newtheorem{claim}[theorem]{Claim}
\newtheorem{question}[theorem]{Question}

\numberwithin{equation}{section}
\numberwithin{theorem}{section}

\begin{document}

\title[A variant of the multi-agent rendezvous problem] 
{A variant of the multi-agent rendezvous problem}

\author{Peter Hegarty, Anders Martinsson, Dmitry Zhelezov} \address{Department 
of Mathematical Sciences, 
Chalmers University Of Technology and University of Gothenburg,
41296 Gothenburg, Sweden} \email{hegarty@chalmers.se, andemar@chalmers.se, 
zhelezov@chalmers.se}



 \keywords{}

\date{\today}

\begin{abstract} 
The classical multi-agent rendezvous problem asks for a deterministic 
algorithm by which $n$ points scattered in a plane can move about at constant 
speed and merge at a single point, assuming each point can use only the 
locations of the others it sees when making decisions and that the visibility 
graph as a whole is connected. In time complexity analyses of such algorithms,
only the number of rounds of computation required are usually considered, not 
the amount of computation done per round. In this paper, we consider 
$\Omega(n^2 \log n)$ points distributed independently and uniformly at random 
in a disc of radius $n$ and, assuming each point can not only see but 
also, in principle, communicate with 
others within unit distance, seek a randomised merging algorithm which 
asymptotically almost surely (a.a.s.) runs in time $O(n)$, in other words in 
time linear in the radius of the disc rather than in the number of points. 
Under a precise set of assumptions concerning the communication capabilities 
of neighboring points, we describe an algorithm which a.a.s. runs in time 
$O(n)$ provided the number of points is $o(n^3)$. Several questions are posed 
for future work.  
\end{abstract}

\maketitle

\setcounter{section}{-1}

\section{Notation}\label{sect:notation}

Let $g, h: \mathbb{N} \rightarrow \mathbb{R}_{+}$ be any two functions. We 
will employ the following notations throughout, all of which are quite 
standard: 
\par (i) $g(n) \sim h(n)$ means that $\lim_{n \rightarrow \infty} 
\frac{g(n)}{h(n)} = 1$.
\par (ii) $g(n) \lesssim h(n)$ means that $\limsup_{n \rightarrow \infty} 
\frac{g(n)}{h(n)} \leq 1$.
\par (iii) $g(n) \gtrsim h(n)$ means that $h(n) \lesssim g(n)$.
\par (iv) $g(n) = O(h(n))$ means that $\limsup_{n \rightarrow \infty} 
\frac{g(n)}{h(n)} < \infty$.
\par (v) $g(n) = \Omega(h(n))$ means that $h(n) = O(g(n))$.
\par (vi) $g(n) = \Theta(h(n))$ means that both $g(n) = O(h(n))$ and 
$h(n) = O(g(n))$ hold.
\par (vii) $g(n) = o(h(n))$ means that $\lim_{n \rightarrow \infty} 
\frac{g(n)}{h(n)} = 0$. 
\\
Now suppose instead that $(g(n))_{n=1}^{\infty}, (h(n))_{n=1}^{\infty}$ are two
sequences of random variables. We write $g(n) \sim h(n)$ if, 
for all $\varepsilon_1, \varepsilon_2 > 0$ and $n$ sufficiently large, 
\begin{equation}
\mathbb{P} \left( 1-\varepsilon_1 < \frac{g(n)}{h(n)} < 1 + \varepsilon_1 \right) > 1-\varepsilon_2.
\end{equation}
Similarly, we write $g(n) \lesssim h(n)$ if, for all $\varepsilon_1,
\varepsilon_2 > 0$ and $n$
sufficiently large, 
\begin{equation}
\mathbb{P} \left( \frac{g(n)}{h(n)} < 1 + \varepsilon_1 \right) > 1 - 
\varepsilon_2.
\end{equation}
We will employ the standard phrase ``asymptotically almost surely (a.a.s.)''
when considering a sequence of events $(A_n)_{n=1}^{\infty}$ such that
$\mathbb{P}(A_n) \rightarrow 1$ as $n \rightarrow \infty$. 
We use the notation $n \gg 0$ to denote
``$n$ sufficiently large''.
\par Finally, if $\bm{p}$ is a point in $\mathbb{R}^2$ and $\varepsilon > 0$,
we denote by $B_{\varepsilon}(\bm{p})$ the open ball of radius $\varepsilon$
about $\bm{p}$.  

\setcounter{equation}{0} 

\section{Introduction}\label{sect:intro}

One day while performing the mathematical equivalent of sitting at a bar, 
drinking beer and philosophising - that is, browsing the day's listings on 
arxiv.org - one of the 
three authors was somehow reminded of a classic scene from the movie 
Terminator 2 \cite{CW}. At one point toward the end of the film, the evil 
Terminator becomes liquified after an explosion. Blobs of liquid metal are 
scattered on the ground and he
appears to have been ``terminated''. However, as if guided by some mysterious 
superior intelligence, the blobs suddenly start moving 
toward one another, eventually merging and reforming the intact Terminator,
who then sets off on the rampage again. 
\\
\par
At the time, we were completely unaware of the extensive literature on the
subject of \emph{rendezvous problems}. 
The movie scene got us thinking, and the 
results presented in this article are the result of those thoughts.
Only afterwards, when we circulated a draft paper to a colleague who is an
expert in computational geometry, were we alerted to the existing literature. 
Having examined the latter, it seems to us that the kind of multi-agent 
rendezvous problem we ended up considering has two important differences from 
what is studied in the literature. 
Firstly, we imagined that agents which could see 
each other could also, in principle, communicate and thus exchange information.
Our agents are thus more anthropomorphised than the robots normally
imagined. Secondly, we took a ``probabilist's approach'', and sought randomised 
algorithms which would work well, for generic configurations of agents, and in 
the limit as both the number and the spatial distribution of agents tended to 
infinity. 
These facts probably render our approach of less practical significance, 
though we will try to convince the reader that it is, at the very least, still 
an interesting thought-experiment and involves some elegant mathematics. 
\\
\par
Rendezvous problems (RPs) originated in work of Steven Alpern, who considered 
the problem faced by two people who are placed randomly in a known, bounded 
search 
region in $\mathbb{R}^2$ and move about at unit speed to find each other in 
the least expected time. Alpern originally imagined that there were a finite 
number of possible meeting points, but later formalised the continuous version 
of the problem, and even generalised it to other compact 
metric spaces \cite{A}. A 
variety of models have been proposed in the case of an arbitrary number of 
agents, which are then usually thought of as autonomous robots, rather than 
fully-fledged humans. An elegant, and  classical model, makes the following 
assumptions:
\\
\\
\emph{RP-1:} Each robot is idealised as a mobile point in $\mathbb{R}^2$. 
Rendezvous means the merging of all the robots at a single point. The robots
are assumed to be identical. 
\\
\emph{RP-2:} The robots can move at constant speed.
\\
\emph{RP-3:} There is some fixed constant $\varepsilon > 0$ such that, at all 
times, each robot can only see those others which are within a 
distance $\varepsilon$ of itself. For a given 
configuration of point robots, this gives rise to a geometric 
graph $G$ whose vertices 
are the points, and where an edge is placed between any two points at distance 
at most $\varepsilon$ from each other. Usually $G$ is called the 
{\em visibility graph} corresponding to the point configuration. 
Any two points connected by an edge we shall refer to as {\em neighbors}. 
\\
\emph{RP-4:} The initial configuration is such that the visibility graph is 
connected. However, the robots have no a priori knowledge 
that they all lie within some specific region of $\mathbb{R}^2$. This is in 
contrast to Alpern's original formulation, where the search region was 
known{\footnote{Note that, given RP-3, any finite number of 
points can almost surely 
merge, even if the visibility graph $G$ consists entirely of 
isolated vertices and there is no a priori knowledge of a search region. 
For example, the points could perform independent Brownian motions, and
since, as is well-known (see \cite{Du} for example), 
the trajectory of a Brownian motion in 
$\mathbb{R}^2$ is almost surely dense, any finite number of points 
will most surely be able to merge in this manner.
However, as is also well-known, the expected running time for this
procedure will always be infinite, even in the case of two points whose
initial separation is any positive constant greater than the visibility range
$\varepsilon$. Also note that the procedure will not work at all in
higher dimensions, as Brownian motions are then no longer almost surely
dense.}}. 
\\
\emph{RP-5:} Robots do not communicate explicitly with each other, even when 
they are neighbors. Each robot is equipped with a sensor which allows it to 
determine the exact locations of all its neighbors at any instant. This is the 
only input it has from its neighbors when deciding how to move.
\\
\par
Under these assumptions, one seeks a deterministic algorithm which guarantees 
merging. It is natural to seek an algorithm where the actions of the robots 
are \emph{synchronised}. What this means is that the algorithm should consist 
of \emph{rounds}. In each round, every robot notes the locations of its 
neighbors, performs some computation and determines a point to which it then 
moves, in a straight line, from its current location. All the robots execute 
these tasks before the next round begins. 
\par The classical solution to this problem was provided by Ando, Suzuki and 
Yamashita \cite{ASY}, see also \cite{AOSY}, 
whose algorithm utilises a well-known concept in 
computational geometry. Let 
$k$ be a positive integer and let $S$ be a set of $k$ points in the plane, 
say $S = \{\bm{p}_1,\dots, \bm{p}_k\}$, where $\bm{p}_i = (x_i,y_i)$. For each 
point $\bm{p} \in \mathbb{R}^2$, let
\begin{equation}
f(\bm{p}) := \max_{1 \leq i \leq k} ||\bm{p} - \bm{p}_i ||,
\end{equation}

where $|| \cdot ||$ denotes Euclidean distance. There is a unique 
point $\bm{p}_{S}$ at which the function $f$ attains its minimum value, namely 
the center of the unique
closed disc of minimum radius which includes all the points of $S$. The point 
$p_S$ is called the \emph{center of the smallest enclosing circle (CSEC)}, 
with respect to the points in $S$. In each round of the ASY-algorithm, a 
robot moves in the direction of the CSEC of the points in its viewing range 
(including itself), but not necessarily to exactly the CSEC because one must 
ensure that the visibility graph remains connected. Indeed, the main technical 
challenge for the algorithm is to specify how to satisfy this constraint. Note 
that, in the last round before merging, all the robots must be visible to one 
another and then they all compute the same CSEC and move to it. By definition 
of the CSEC, and given RP-2, this is the optimal solution to the 
multi-agent Rendezvous Problem when there is global visibility.
\par We found a number of papers which investigated the time complexity of
the ASY and similar rendezvous algorithms, see for example \cite{MBCF} and 
\cite{De}. Two things struck us about these analyses. Firstly, the complexity 
is defined in terms of the number of rounds. The time taken to execute 
computations within each round is not considered. Secondly, the analyses tend 
to focus on worst-case scenarios, in other words, the primary goal is to 
obtain an upper bound for the number of rounds which is valid for any 
connected visibility graph. Hence the bounds tend to be expressed in terms of 
the number of robots. A remark on page 2221 of \cite{MBCF} hints that, if one 
begins by assuming that the robots are confined to some fixed bounded region 
in $\mathbb{R}^2$, then for ``generic'' initial configurations, the size and 
shape of the region may be what really determine the time complexity rather 
than the number of points. This viewpoint means abandoning RP-4, of 
course, though note that it is less restrictive than in Alpern's original 
formulation, since one is not assuming that the initial search region is known 
to the agents, only that such a bounded region exists. 
We have not seen any paper which follows up on this 
idea, however. 
\\
\par
As already stated, we began thinking about Rendezvous-type problems well 
before we became aware of the existing literature on the 
subject{\footnote{Indeed, we didn't even know the term ``rendezvous problem'', 
and had instead invented our own term \emph{terminator problem}, in 
acknowledgement of our source of inspiration.}}. With hindsight, the point of 
view we took addresses the two issues raised above about existing 
time-complexity analyses. However, to do so and yet obtain the result we were 
looking for, we had to abandon RP-5, and allow our robots certain 
abilities to communicate with their neighbors. Furthermore, we had to foucs on
asymptotic results and randomised rather than fully deterministic algorithms.
We now describe our viewpoint.
\par The simplest kind of closed, bounded region in $\mathbb{R}^2$ is a disc, 
so suppose the points are initially confined to a disc of radius $n$. The disc 
radius is our primary parameter, and we will be interested in asymptotic 
results, hence in letting $n \rightarrow \infty$. We are also interested in 
``generic'' configurations, so we imagine that each point is placed uniformly 
at random in the disc, and independently of all other points. It is 
well-known - 
see \cite{P} for example - that if the number $N$ of points satisfies 
$N = \Omega(n^2 \log n)$, then the visibility graph will a.a.s. be connected, 
whereas if $N = o(n^2 \log n)$, then it will a.a.s. be disconnected. So we 
suppose the former holds. If we run the ASY-algorithm, then in the last step 
each robot will have to compute the CSEC for some set of $N$ points in the 
plane. Trivially, any procedure for doing this will require time $\Omega(N)$ 
to execute. Hence, if we include this time in our analysis, we have a trivial 
lower bound of $\Omega(n^2 \log n)$ for 
the time to rendezvous. Taking all the rounds of the procedure into account 
will make things even worse, though perhaps not significantly since
\par (a) during earlier rounds, each robot will have far fewer than $N$ 
neighbors, hence most of the computations are done only towards the end. Since 
the average density of points is $\Omega( \log n)$ initially, and a robot
moves distance $O(1)$ per round, what we can be 
certain of is that the extra contribution to the running time 
coming from earlier rounds is $\Omega(n \log n)$.
\par (b)  
it is 
a classical result that computation of a CSEC can indeed be solved 
in time linear in the 
number of points. Megiddo \cite{M} was the first to describe a deterministic 
linear-time algorithm. Later, Welzl \cite{W} developed an 
alternative procedure which is considered the state-of-the-art solution. His 
algorithm is probabilistic, with expected linear running time, but it is much 
simpler to describe and to implement that Megiddo's.
\par
To repeat, given that the initial visibility graph $G$ is a generic, 
connected geometric graph inside a disc of radius $n$, the actual running time 
of the ASY-algorithm under conditions RP 1-5 above is $\Omega (n^2 \log n)$. 
On the other hand, if the points somehow ``knew'' where to rendezvous, then 
RP-2 means they could do so in time 
$O(n)$. The goal of our study was to find a probabilistic algorithm which 
a.a.s. led to rendezvous in time $O(n)$. It seems clear, though we have not
proven it rigorously (see Section \ref{sect:discuss}), 
that this is impossible under RP-5. So we are 
forced to allow some kinds of communication between points, which is perhaps a 
more radical departure from the classical multi-agent Rendezvous Problem than 
some a priori assumption about the geometry of the point configuration. Part of
the goal then becomes to find an algorithm which contains as simple and as few 
as possible assumptions about how neighboring points can communicate. This is 
admittedly a subjective requirement, but as long as we can make all 
assumptions precise we will at least have a well-defined mathematical problem.
There is also a conceptual problem with, say, the ASY-algorithm as $n 
\rightarrow \infty$, namely that each robot would need unlimited capacity
to store the results of the computation of a CSEC. We will thus also
have a preference for algorithms which overcome this problem. 
\par Section \ref{sect:heart} is the heart of the paper. Having formalised our 
assumptions about how points can communicate we will prove our main result, 
Theorem \ref{thm:main}. Informally, it asserts that there is a randomised 
rendezvous algorithm which a.a.s. runs in time $O(n)$ provided the initial 
configuration of points is not \emph{too} dense. Significantly, our algorithm
consists of two main steps, the first and more difficult of which involves the
choice of a \emph{leader} amongst the points. Indeed, it may make more
sense to think of our algorithm as being for this purpose, rather than for the 
purpose of rendezvous. We will have a lot more to say about this matter in 
Section \ref{sect:discuss}, which contains a critical analysis of our 
alternative assumptions, plus suggestions for further developing the ideas of 
this paper. A sceptical reader may choose to only skim through Section 
\ref{sect:heart}, ignoring detailed proofs, before reading Section 
\ref{sect:discuss}.    

\section{Rendezvous in a disc with local communication}\label{sect:heart}

We start with two lemmas which will be used in the proof of the main theorem
below.

\begin{lemma}\label{lem:graph}
Let $f: \mathbb{N} \rightarrow \mathbb{N}$ be an increasing function. For 
each $n \in \mathbb{N}$ suppose $f(n)$ points are placed uniformly
and independently in the interior of a disc $\mathcal{D}$ 
in $\mathbb{R}^2$ of radius $n$. 
Let $G = G_n$ be the geometric graph whose vertices are these points, and
with an edge between any two points the Euclidean distance between whom is
at most one. Set $\lambda_n := \frac{f(n)}{\pi n^2}$, the average density
of points in the disc. There is an absolute constant $C_1 > 0$ such that, if 
$\lambda_n > C_1 \log n$, then a.a.s.
\par (i) the degree of every vertex in $G$ lies between $\frac{\pi}{3} 
\lambda_n$ and $2\pi\lambda_n$,
\par (ii) the graph diameter of $G$ is at most $6n$.
\end{lemma}

\begin{proof}
There is a rich literature on random geometric graphs - see \cite{P} - and the 
above statements certainly follow from well-known
(stronger) ones. We sketch proofs here for the sake of completeness. 
Choose $n$ and let $\bm{p}$ be any point in $\mathcal{D}$. 
The number of vertices of $G_n$ which are at distance
less than one from $\bm{p}$ is distributed as Bin$\left( f(n), a \lambda_n 
\right)$, where $a$ is the area of $B_{1}(\bm{p}) \cap \mathcal{D}$. Hence
$a \in \left( \frac{\pi}{2}, \pi\right]$, depending on how close
$\bm{p}$ is to the boundary of $\mathcal{D}$. 
It follows that the probability that the degree of any particular 
vertex of $G_n$ lies outside
the interval $\left( \frac{\pi}{3} \lambda_n, 2\lambda_n \right]$ 
is at most $e^{-C_2 \lambda_n}$, 
for some absolute constant $C_2 > 0$. Thus (i) follows by a simple union bound.
Indeed, note that (i) holds if we replace the factors $\frac{1}{3}$ and $2$
by any constants less than $\frac{1}{2}$ and greater than $1$ respectively. 
\par To prove (ii), consider two vertices $\bm{p}$ and $\bm{q}$ in $G_n$. 
No matter where they are located in $\mathcal{D}$, we can connect them by 
a piecewise rectangular tube of width $\frac{1}{3}$ and of length at most 
$2n$ (see Figure 1).
We can divide this tube into square sections of side-length $\frac{1}{3}$ 
and observe that
any two vertices located in adjacent sections are joined by an edge in 
$G_n$. Hence, if no section of the tube is empty, then the graph distance
from $\bm{p}$ to $\bm{q}$ is at most $6n$. 
The probability that any section is completely empty of vertices 
is at most $e^{-C_3 \lambda_n}$ for some absolute constant $C_3 > 0$. 
The claim of part (ii) follows by some simple union bounds. 
\end{proof}

\begin{figure*}[ht!]
  \includegraphics[]{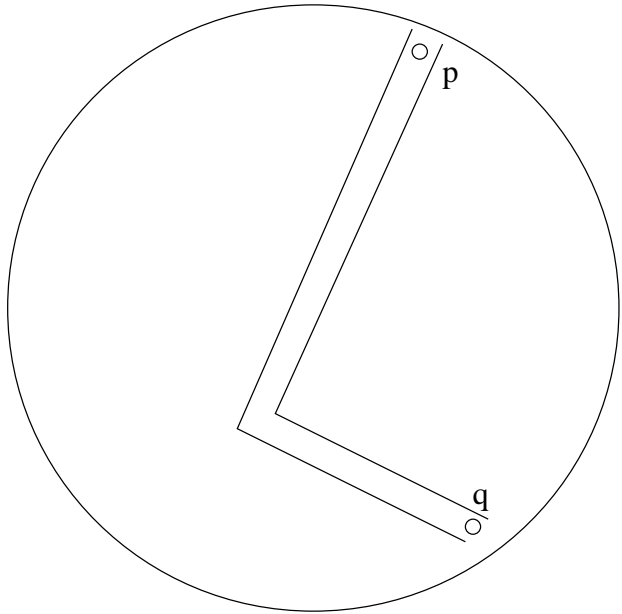} 
 \label{fig:ini}
\caption{A tube connecting two points in the disc.}
\end{figure*}

\begin{lemma}\label{lem:read} 
Let $n \in \mathbb{N}$ and let $k = k(n)$ be a positive
integer such that $k \rightarrow \infty$ as $n \rightarrow \infty$. 
Suppose $n$ binary strings $\bm{s}_1,\dots, \bm{s}_n$, each of length $k$,
are generated independently and uniformly at random. Consider a fixed 
string $\bm{s}_i$. Suppose it compares itself with each other string $\bm{s}_j$,
$j \neq i$,
by reading each string bit-by-bit from left to right until it encounters a bit
where the string differs from itself (if $\bm{s}_j = \bm{s}_i$ then the whole
of $\bm{s}_j$ is read). Let $X_i$ be the total number of bits read by 
$\bm{s}_i$ and let $X = \max_{i} X_i$. Then $X \lesssim 2n$.
\end{lemma}

\begin{proof}
$X_i = \sum_{j \neq i} X_{ij}$, where $X_{ij}$ is the number of bits of 
$\bm{s}_j$ that are read by $\bm{s}_i$. We have 
\begin{eqnarray}
\mathbb{E}[X_{ij}] = \sum_{t=1}^{k} t \cdot 2^{-t} \sim 2, \\
\mathbb{E}[X_{ij}^{2}] = \sum_{t=1}^{k} t^2 \cdot 2^{-t} \sim 6.
\end{eqnarray}
Hence $\mathbb{E}[X_i] \sim 2n$ and since, for a fixed $i$, 
the $X_{ij}$ are independent, 
Var$(X_i) \sim 2n$ also. It follows from the Central Limit Theorem that,
for any $\varepsilon > 0$, there exists $c_{\varepsilon} > 0$ such that
$\mathbb{P}(X_i > (2+\varepsilon)n) < e^{-c_{\epsilon} n}$. The lemma then follows
from a union bound.  
\end{proof}

In order to be able to prove a rigorous mathematical result, Theorem 
\ref{thm:main} below,
we need to 
specify precisely our assumptions about the
capabilities of the points that are tasked with merging. A reader
may variously complain that the following list of 11 assumptions is either
too long, too ad hoc or dubious from the point of view of physics. We accept 
such criticisms as valid, but nevertheless think the thought
experiment we are engaged in is worth pursuing, for two
reasons. On the one hand, we think our result is \emph{mathematically}
interesting since, as will become clear from the proof of Theorem 
\ref{thm:main}, our
algorithm very much relies on the geometry of the point 
configuration. On the other hand, none of our 11 assumptions seems
completely ridiculous from an anthropomorphic standpoint, that is, if we
imagine our ``points'' as being closer to intelligent human agents rather than
much more primitive robots. The fact that our algorithm involves choosing a 
leader gives greater justification for this anthropomorphic viewpoint. One
suggestion to a potential reader is to skip the list of assumptions and
go directly to Theorem \ref{thm:main} 
and its proof, returning to the list only if one 
encounters anything in our algorithm that seems completely implausible. 
Otherwise, we shall return to this discussion in Section \ref{sect:discuss}. 
\\
\par 
Henceforth, we assume a fixed choice of length and time units.
All implicit constant factors in the following list, and indeed throughout the
remainder of Section \ref{sect:heart}, are positive and 
absolute. 
We consider the first five assumptions to be minimal requirements if we wish 
to adhere to the spirit of the classical multi-agent RP, as expressed by
RP 1-5, but allow local communication in principle.  
\\
\\
{\em Assumption 1:}\label{ass:one} Each point 
can move at speed at most one. That this constant is the same for all 
points corresponds, intuitively, to the assumption that all the points are 
identical, and in particular have the same ``mass''. Consequently, we could 
further assume that if, at some stage of the merging process, two points 
became stuck together, then this new ``heavier point'' could subsequently only 
move at some maximum speed less than one. However, since in our algorithm 
a leader is first chosen and then all the other points move to it, 
this last issue is not relevant. See Questions \ref{quest:twostep} and
\ref{quest:intermediate} below, however. 
\\
\\
{\em Assumption 2:} Each point can only directly communicate (by 
transmission or reception of electromagnetic or chemical signals, say, though
the precise details don't matter) with other 
points that are within a distance one of itself. As in Section 
\ref{sect:intro}, two points which can communicate directly will be called 
neighbors. Henceforth, what was referred to in Section \ref{sect:intro} as
the visibility graph will now be called the \emph{communication graph}. 
\\
\\
{\em Assumption 3:} All signals travel at a fixed speed, which we 
can think of as being determined by the laws of physics (and maybe chemistry).
\\
\\
{\em Assumption 4:} If a point broadcasts a signal, then any neighbor
which receives it can locate, in time $O(1)$, 
the exact point source of the broadcast.
\\
\\  
{\em Assumption 5:} Each point is \emph{immortal} 
in the sense that it can carry 
on the various activities desired of it at the 
same constant rate indefinitely. Note that this assumption is necessary to 
even make sense of any asymptotic result in which the diameter of the initial 
configuration tends to infinity. Here it is natural to think of the points as 
representing biological agents, who can ``eat'' to replenish their energy 
stores. 
\\
\par
Our next two assumptions regard the types of signals that a point 
can broadcast and how they are processed. 
We assume there are two types of signals.  
\\
\\
{\em Assumption 6:} Each point can broadcast simoultaneously to all its
neighbors. There are a bounded, if perhaps large, number of such signals, 
which we will think of as being different 
\emph{colors}. Any color signal can be generated,
turned off
or recognised in time $O(1)$.  
In addition, we assume that every point can both 
\emph{isolate} and \emph{filter} received 
color signals. By isolating we mean that, 
if a point is receiving signals in the same color
from multiple sources, then it can identify and locate individual sources
at a rate of $\Omega(1)$ sites per time unit. By filtering we mean that, 
if some set $\mathcal{C}$ of colors are presently being broadcast
amongst a point's neighbors, but it is only interested in some subset 
$\mathcal{C}^{\prime}$ of colors, then in time $O(1)$ it can ``put on an 
appropriate pair of goggles'' and henceforth scan only for colors 
in $\mathcal{C}^{\prime}$, and not
be disturbed in any way by interference from colors in $\mathcal{C}
\backslash \mathcal{C}^{\prime}$.   
\\
\\
{\em Assumption 7:} Each point can also generate single bits and 
send them to individual neighbors. A bit can be generated, an
individual neighbor identified, and the bit sent to the
neighbor all in time $O(1)$. The receiving point can process the bit and
(see Assumption 4) identify its exact source in time $O(1)$.

\begin{rem}\label{rem:sixpower}
Assumption 6, together with the last sentence of 
Assumption 11 below, 
is a powerful tool. To get a feeling for this, suppose all the 
points are initially inside an unspecified disc of diameter one, so we have
global visibility but no universal point of
reference. Then, for $N$ points, a.a.s. they could rendezvous
in time $O(\log N)$. For example, suppose each point is red by default. 
Each point can start generating random bits, and turn blue as soon as it
generates a zero. If, at some step, all the remaining red points turn blue, 
then these ``last men standing'' go back to red and try again. 
It is clear that, almost surely, within $O(\log N)$ rounds there will be 
a time at which exactly one point is red. 
This point then becomes the \emph{leader}. Since it
is visible to all other points, they can all merge at its location. By
comparison, if we only assumed RP 1-5, then the computation of a CSEC would
require time $\Omega(N)$, and that is only if we ignore
the problem of storage capacity.
\end{rem}    

Finally, we make explicit four additional assumptions which our algorithm
will exploit. They concern the abilities of the points to perform certain tasks
or manouevers. Note that Assumption 9 will only be used to deal with the
fact that each point has a finite storage capacity - see Remark
\ref{rem:storage} below. 
\\
\\
{\em Assumption 8:} If one point receives a color 
signal from another and then the sender, after waiting time $O(1)$ so
that the receiver locates it, starts to move while still
transmitting the color, then the receiver can
\emph{track} its movements in real time, and hence follow it if necessary
in such a way that the vector separation 
between the two points remains constant. If the sender remains stationary, then 
the receiver can move toward it in a straight line at speed one. 
\\
\\
{\em Assumption 9:} Given three points $\bm{p}_1$, $\bm{p}_2$, $\bm{p}_3$ which
are pairwise mutual neighbors, $\bm{p}_1$ can \emph{point out} $\bm{p}_3$ 
to $\bm{p}_2$ and $\bm{p}_2$ can process this information, all within time 
$O(1)$. We can imagine the actual ``pointing'' being done by, for example, 
$\bm{p}_1$ shining a laser at $\bm{p}_3$. A more ``low-tech'' solution would be
for $\bm{p}_1$ to walk toward $\bm{p}_3$, do a little dance around it and then
walk back to where it came from. This would work given Assumption 8, and
also assuming each point can leave a \emph{beacon} (which transmits a 
color reserved for beacons), so that it can return to its exact 
starting point after a walk within a radius of one.   
\\
\\
{\em Assumption 10:} Given any fixed $\varepsilon \in (0, 1]$, a point can 
\emph{sweep}
its $\varepsilon$-neighborhood in such a way that it identifies the
individual 
points in it at a rate of $\Omega(1)$ points per time unit and does not
miss any points. In fact, our
algorithm will only require this capacity for some finite 
number of predetermined values of $\varepsilon$, namely $\varepsilon 
\in \left\{\frac{1}{10C_{15}}, \frac{1}{2}, 1\right\}$, 
where $C_{15}$ is an absolute constant that will
appear in the proof of Theorem \ref{thm:main}. 
Note that, what we're assuming here is more than just the ability
to determine, in time $O(1)$, whether 
an identified neighbor is within a distance $\varepsilon$ or not. Rather,
we are assuming that all neighbors at distance greater than $\varepsilon$ can
be blotted out as the scan is performed, perhaps by tuning a receiver so
that it blocks out signals which are weaker than a certain threshold.
We're also assuming that a point can
physically perform a sweep, for example by rotating on an axis. 
\\
\\
{\em Assumption 11:} Following on from the previous assumption, given 
a point $\bm{p}$, a value of $\varepsilon$ and any ninety-degree sector of
$B_{\varepsilon}(\bm{p})$, the point can determine in time
$O(1)$ whether that sector is empty of neighbors or not. Also, given a color,
it can determine in time $O(1)$ whether or not anyone inside 
$B_{\varepsilon}(\bm{p})$ is currently broadcasting
that color. 
\\
\par
This completes our list of asssumptions. One thing that might strike the 
reader as surprising is that we do not require the points to 
possess clocks, for example for the purpose of synchronising their actions.
In fact, as we shall see, 
our algorithm only requires a very rudimentary synchronisation
which can be achieved without perfect clocks, given the above assumptions.
We are ready to state our main result. 

\begin{theorem}\label{thm:main}
There is a randomised algorithm $\mathcal{A}$ for which the following
holds: There are absolute positive constants 
$C_4, C_5$ such that, if $f: \mathbb{N} \rightarrow \mathbb{N}$ is a function 
satisfying $C_4 n^2 \log n < f(n) = o(n^3)$ then,
under Assumptions 1-11 above, 
if $f(n)$ points are
placed uniformly and independently at random in the interior of a closed
disc $\mathcal{D} = \mathcal{D}_n$ of radius $n$ in $\mathbb{R}^2$ 
and proceed to execute the 
algorithm $\mathcal{A}$, they will a.a.s. merge in time at most $C_5 n$. 
\end{theorem}

\begin{proof}
Throughout the proof there will appear a sequence of absolute, positive
constants which will be denoted by $C_i$, for $i = 6,7,\dots, 16$. Some of these
constants will be related to one another, as well as to the constants 
$C_4$, $C_5$ appearing in the statement of the theorem. We will not bother 
with making these dependencies explicit, however.
\par 
As in Lemma \ref{lem:graph} we will denote the average point density by 
$\lambda_n$, i.e.: $\lambda_n = \frac{f(n)}{\pi n^2}$. We will denote the 
boundary of the disc $\mathcal{D}$ by $\delta \mathcal{D}$. 
\par Our algorithm $\mathcal{A}$ will consist of two main steps:
\\
\\
{\sc Step 1:} Choose a \emph{leader}.
\\
{\sc Step 2:} Everyone move to the leader's location. 
\\
\par
More precisely, the idea is that, in Step 1, the $f(n)$ points should perform
a sequence of operations at the end of which, a.a.s. exactly one
point will be 
in possession of information which identifies it as ``leader'', whereas
every other point will possess information which allows them to 
rule themselves out as leader. In fact, Step 2 can begin
once one point believes itself to be the leader, the crucial thing being
that a.a.s. no other point will subsequently reach the same conclusion and 
muddy the waters. Step 1 is the trickier part of the 
algorithm and we will go into the details below. 
\par In Step 2, the leader signals
its identity to all its neighbors by broadcasting the color red, this being
the color reserved specifically for the signal ``I am the leader'' 
(see Assumption 6). 
Once a point receives a red signal and identifies its source, 
it broadcasts red in turn to all its 
neighbors and then moves towards the location from which it received the 
signal. If
it received a red signal from multiple sources, it chooses one of these at 
random and follows it. Our various assumptions (in particular, Assumption 8) 
guarantee that Step 2 will result in 
all $f(n)$ points merging at the leader's location in time $O(n)$, if and
only if the original red broadcast reaches all $f(n)$ points in time $O(n)$. 
It follows from Lemma \ref{lem:graph}(ii) that this will a.a.s. be the case, 
provided $f(n) > C_6 n^2 \log n$ for a sufficiently large $C_6$.
Note, in particular, that Step 2 works for arbitrarily dense configurations
of points, the upper bound on $f(n)$ in Theorem \ref{thm:main} is not needed. 
\\
\par
It thus remains to describe the implementation of Step 1. 
There are two crucial ingredients and both rely on
the upper bound $f(n) = o(n^3)$. 
\par The first ingredient is that, if
$f(n) = o(n^3)$ then, by Lemma \ref{lem:graph}(i), a.a.s. every point has 
$o(n)$ neighbors. Thus, by Assumption 10, every point can count its
neighbors 
in time $o(n)$. There is one subtlety here, however.  
Recall from Section \ref{sect:intro} that we prefer to think of 
each point as possessing 
a finite memory capacity. Hence, for $n \gg 0$, it will in general
not be able to store internally the result of a count of its
neighbors. We'll present our solution to this problem further down, but it
makes our algorithm a good deal more clunky than it would be otherwise. 
See also Remark \ref{rem:storage}. 
\par The second crucial ingredient is that the ``boundary'' of the 
communication 
graph will a.a.s. contain $O(n)$ points. We need a precise
definition:

\begin{defi}\label{def:boundary} 
Let $G$ be the communication graph of a collection of
points distributed in a disc. Let $v$ be one of these points. If $v$ has
a set of neighbors $v_1,\dots, v_k$ such that, setting $v_{k+1} := v_1$, 
\par (i) $v_i$ and $v_{i+1}$ are also neighbors, for each $i = 1,\dots, k$,
\par (ii) reading clockwise, $v_1,\dots, v_k$ are the vertices of a simple
 polygon in the plane which encloses $v$,
\\
then $v$ is said to be a \emph{non-boundary point} of $G$. If no such 
set of neighbors of $v$ exists, we say it is a \emph{boundary point} of $G$.
We denote by $\delta G$ the set of boundary points of $G$.  
\end{defi}

\begin{claim}\label{clm:boundsize} 
Suppose $f(n)$ satisfies the assumptions of Theorem 
\ref{thm:main}. Then
\par (i) there are absolute positive constants $C_7, C_8$ such that, a.a.s.,
$C_7 n \leq |\delta G| \lesssim C_8 n$. 
\par (ii) a.a.s. every point in $\delta G$ is at distance at most $C_{9}
\frac{\log n}{\lambda_n}$ from
$\delta \mathcal{D}$, where $C_9$ is another absolute constant.
\end{claim}

{\em Proof of Claim.} These kinds of statements may also follow from
``well-known facts'' about geometric graphs, but we choose to give
complete proofs. 
\par For a given $n$, let $\mathcal{D}^1$ denote the closed 
disc with the
same center as $\mathcal{D}$, but of radius $n - \frac{1}{n}$. Set
$\mathcal{D}^2 := \mathcal{D} \backslash \mathcal{D}^1$ and, for
each $i = 1, 2$, let $\delta_i G := \delta G \cap \mathcal{D}^i$. 
\par Let $\bm{p}$ be a point of $\mathcal{D}$, and let $r$
denote its distance from $\delta \mathcal{D}$. Let 
$\bm{o}$ denote the center of $\mathcal{D}$ and let $\bm{p}_1 \bm{p}_2$ be the
chord through $\bm{p}$ which is perpendicular to the line through
$\bm{o}$ and $\bm{p}$ (see Figure 2 on page 10). It is easy to check that 
$||\bm{p} \bm{p}_i|| \geq \frac{1}{2}$, for $i = 1, 2$, 
whenever $r > \frac{1}{n}$ and $n \gg 0$. 

\begin{figure*}[ht!]
  \includegraphics[]{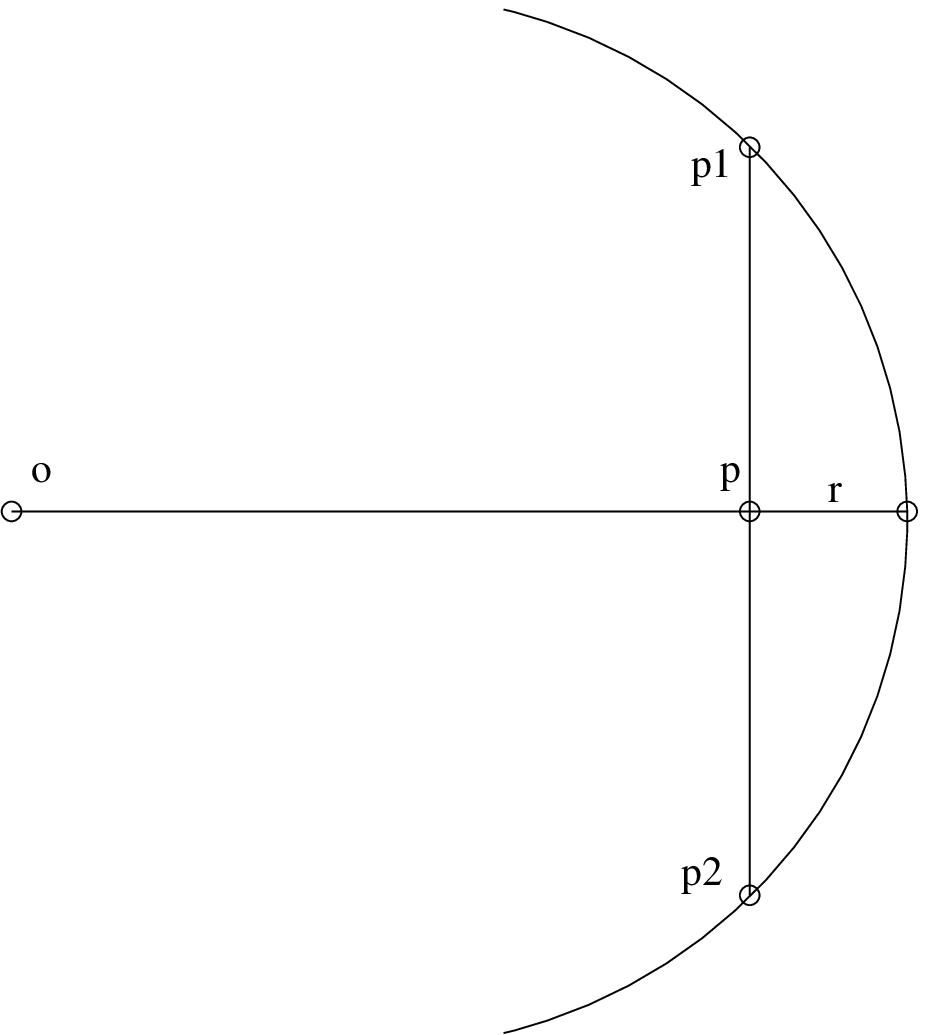} 
 \label{fig:ini}
\caption{}
\end{figure*}

\par First consider points 
in $\mathcal{D}^2$. The total number of 
such points in the configuration is distributed as 
Bin$\left(f(n), \frac{a}{\pi n^2} \right)$, where $a$ is the area of 
$\mathcal{D}^2$, 
hence $a \sim 2\pi$. It follows that a.a.s. the number of such points is 
$o(n)$, and hence
that $|\delta_2 G| = o(n)$ also.     
\par Next consider a point $\bm{p}$ in the configuration at distance
$r > \frac{1}{n}$ 
from $\delta \mathcal{D}$. If $\bm{p}$ is to lie in $\delta G$, then
at least one of the regions $\mathcal{R}_i$, $i = 1,\dots, 4$, in Figure 3 
on page 11 must be empty. 
By a simple union bound, the probability of this is at 
most 
$4 \cdot \exp (-C_{10} \lambda_n \cdot \max\left\{r, \frac{1}{2}\right\} )$, 
for some $C_{10} > 0$. Part (ii) of the 
claim now follows. Note that this immediatetely implies in turn the 
lower bound in part (i), since
$\delta G$ must be connected, as an induced subgraph of $G$. For the upper
bound in
part (i), let us consider the random variable $X$ which is the number
of points in the configuration 
at distance $r > \frac{1}{n}$ from $\delta \mathcal{D}$ for which
at least one of the four regions in Figure 3 is empty. By definition, 
$X$ stochastically dominates $|\delta_1 G|$, so it suffices to prove that
$X \lesssim C_{11} n$. For the first moment 
we immediately have an upper bound
\begin{equation}\label{eq:expect}
\mathbb{E}[X] \leq C_{12} \lambda_n \int_{0}^{n}
2\pi (n-r) e^{-C_{10} \lambda_n r} \; dr \leq C_{13} n.
\end{equation}

\begin{figure*}[ht!]
  \includegraphics[]{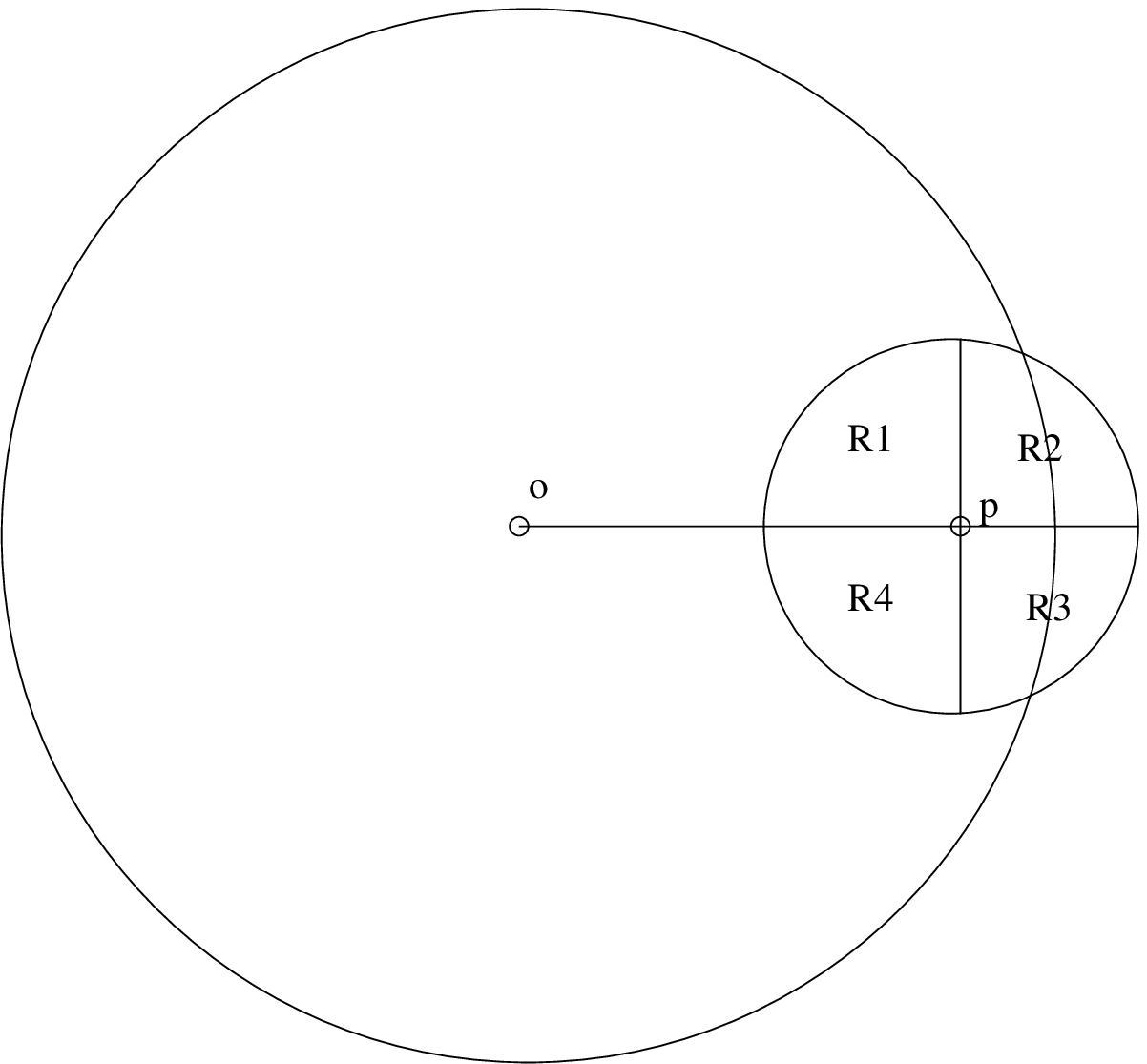} 
 \label{fig:ini}
\caption{The diameter of the smaller disc equals one, so that any two points 
inside it are neighbors.}
\end{figure*}

We will now show that Var$(X) = O(n \log n)$, which will
suffice to complete the proof. We can write $X = \sum X_{\bm{p}}$, a sum
of indicator variables, one for each of the $f(n)$ points in the configuration.
Let $A_{\bm{p}}$ denote the event for which $X_{\bm{p}} = 1$, that is, 
dist$(\bm{p}, \delta \mathcal{D}) > \frac{1}{n}$ and at least one of the
four regions in Figure 3 is empty. There are $O(n^4 \lambda_{n}^{2})$ pairs
of distinct points $\{\bm{p}, \bm{q}\}$, so it suffices to prove that, 
for any pair of points, 
\begin{equation}\label{eq:small}
\mathbb{P}(A_{\bm{p}} \wedge A_{\bm{q}}) - \mathbb{P}(A_{\bm{p}}) \cdot 
\mathbb{P}(A_{\bm{q}}) = O \left( \frac{\log n}{\lambda_{n}^{2} n^{3}}
\right).
\end{equation}
First of all, let us define three auxiliary events $\mathcal{W}_i$, 
$i=1, 2, 3$. Let $\mathcal{W}_1$ be the event that the Euclidean 
distance between $\bm{p}$ and $\bm{q}$ is at least two. Let 
$\mathcal{W}_2$ be the event that at least one of $\bm{p}$ and $\bm{q}$ is at 
least $10C_{9} \frac{\log n}{\lambda_n}$ from $\delta \mathcal{D}$. 
Finally, let $\mathcal{W}_3$ be the complement of 
$\mathcal{W}_1 \cup \mathcal{W}_2$. A priori, we can decompose the 
left-hand side of (\ref{eq:small}) as 
\\
\begin{equation}\label{eq:decomp}
\left[ 
\mathbb{P}(A_{\bm{p}} \wedge A_{\bm{q}} | \mathcal{W}_1) \cdot 
\mathbb{P}(\mathcal{W}_1) - \mathbb{P}(A_{\bm{p}}) \cdot 
\mathbb{P}(A_{\bm{q}}) \right] + 
\mathbb{P}(A_{\bm{p}} \wedge A_{\bm{q}} | \mathcal{W}_2) \cdot 
\mathbb{P}(\mathcal{W}_2) + 
\mathbb{P}(A_{\bm{p}} \wedge A_{\bm{q}} | \mathcal{W}_3) \cdot 
\mathbb{P}(\mathcal{W}_3).
\end{equation}
\par The main idea here is that, if $\mathcal{W}_1$ occurs, 
then the events $A_{\bm{p}}$ and
$A_{\bm{q}}$ are almost negatively correlated. Intuitively, 
if say $A_{\bm{p}}$ occurs, what
we then know is that at least one of the four regions in Figure 3, for the
point $\bm{p}$, is empty. In other words, we just know that some region of 
$\mathcal{D}$ is empty, which must make it less likely that any disjoint
region is also empty. Since $\bm{q}$ is at distance greater than two from
$\bm{p}$, each of its four associated regions is disjoint from the 
corresponding regions for $\bm{p}$. There is one small flaw with this
reasoning, namely the knowledge that $\bm{q}$ is far away from $\bm{p}$
slightly increases the probability of the event $A_{\bm{p}}$ to begin with, and 
vice versa. Namely, we have already ruled out $\bm{q}$ as a close
neighbor of $\bm{p}$, and hence when calculating the probability of 
$A_{\bm{p}}$ we can imagine starting with a configuration of $f(n)-1$ 
rather than $f(n)$ points. In summary, 
\begin{equation}
\mathbb{P} (A_{\bm{p}} \wedge A_{\bm{q}} | \mathcal{W}_1 ) \leq 
\mathbb{P}(A^{\prime}_{\bm{p}}) \cdot \mathbb{P} (A^{\prime}_{\bm{q}}),
\end{equation}
where $A^{\prime}_{\bm{p}}$ is the same event as $A_{\bm{p}}$, but
calculated with respect to an intitial configuration of $f(n)-1$ points
(and similarly for $\bm{q}$). Clearly, 
\begin{equation}
\mathbb{P}(A^{\prime}_{\bm{p}}) \leq \mathbb{P}(A_{\bm{p}}) \cdot \left(
1 - \frac{a}{\pi n^2} \right)^{-1},
\end{equation}
where $a$ is area of any of the four regions in Figure 3, 
hence $a = \pi/4$. Hence,
\begin{equation}\label{eq:first}
\mathbb{P} (A_{\bm{p}} \wedge A_{\bm{q}} | \mathcal{W}_1 ) \leq 
\mathbb{P}(A_{\bm{p}}) \cdot \mathbb{P} (A_{\bm{q}}) \cdot 
\left( 1 + \frac{1}{2n^2} + O \left( \frac{1}{n^4} \right) \right).
\end{equation}
But recall that $\mathcal{W}_1$ is the event that the distance between 
$\bm{p}$ and $\bm{q}$ is at least two, so clearly 
$\mathbb{P}(\mathcal{W}_1) \leq 1 - \frac{2}{n^2}$ (the right constant is $2$
rather than $2^2 = 4$ because we have to consider points close to $\delta
\mathcal{D}$). It follows that the first square-bracketed term in
(\ref{eq:decomp}) will be negative for all $n \gg 0$, 
and it just remains to bound the 
positive contributions coming from
$\mathcal{W}_2$ or $\mathcal{W}_3$. The 
definition of $\mathcal{W}_2$ immediately implies, however, by the same
argument used to obtain part (ii) of the Claim, that it contributes 
negligibly. Hence, it remains to show that
\begin{equation}
\mathbb{P} (A_{\bm{p}} \wedge A_{\bm{q}} | \mathcal{W}_3) \cdot 
\mathbb{P}(\mathcal{W}_3)
= O \left( \frac{\log n}{n^3 \lambda_{n}^{2}} \right).
\end{equation} 
But $\mathcal{W}_3$ occurs if and only if $\bm{p}$ is placed in a strip of 
width $O \left(\frac{\log n}{\lambda_n}\right)$ inside 
the disc boundary and then $\bm{q}$ is subsequently also placed inside this
strip and at distance $O(1)$ from $\bm{q}$. Since the points are placed
independently and uniformly at random, 
\begin{equation}
\mathbb{P}(\mathcal{W}_3) = O \left( \frac{n \frac{\log n}{\lambda_n}}{n^2} 
\times \frac{\frac{\log n}{\lambda_n}}{n^2} \right) = 
O \left( \frac{\log^{2} n}{n^3 \lambda_{n}^{2}} \right).
\end{equation}
This already suffices to prove that Var$(X) = O (n \log^{2} n)$ and
hence to prove part (i) of the Claim. To get rid of another $\log n$ factor, 
it suffices to show that $\mathbb{P} (A_{\bm{p}} \wedge A_{\bm{q}} | \mathcal{W}_3)
= O \left( \frac{1}{\log n} \right)$, and hence to show that
$\mathbb{P} (A_{\bm{p}}) = O \left( \frac{1}{\log n} \right)$, 
conditioned on the assumption that $\bm{p}$ is at distance at most $10C_{9} 
\frac{\log n}{\lambda_n}$ from the disc boundary. But there are a.a.s. 
on the order of
$n \log n$ points in this annulus and we already know that there are a.a.s. 
order $n$ points in $\delta G$, so we are done.  
\\
\par 
We now return to the description of Step 1, which we break down into
three substeps:
\\
\\
\emph{Step 1A:} Each point performs, in time $O(1)$, 
some tests in an attempt to rule out that it belongs to $\delta G$. If all these
tests fail, then it turns blue to signal ``I believe I might belong to 
$\delta G$''. The idea here is that each point does something very similar
to checking the four regions as in Figure 3, and turns blue if and only if at 
least
one of those four regions is empty. A point cannot test exactly this condition,
since it does not know where the centre $\bm{o}$ 
of the disc is. However it can, 
for example, fix some large number $K$, rotate a half turn and,
at equal intervals of $\pi/K$, scan the four quadrants in its 
$\frac{1}{2}$-neighborhood, 
as seen from its current orientation. 
Each time it scans, it just 
wants to decide if each of the four quadrants is empty or not, and
Assumption 11 implies that this can be done in time $O(1)$. Hence all $K$ 
tests can be performed 
in time $O_{K}(1)$. It is clear that, for a sufficiently large but now 
fixed $K$,   
the set of blue points will have the properties
identified in Claim \ref{clm:boundsize}, that is, there will at least 
$C_{14}n$ and a.a.s. at most $C_{15}n$ blue points. Furthermore, for
$C_4 \gg 0$, 
the blue points will a.a.s. all lie within distance $\frac{1}{100 C_{15}}$, say, 
of the 
disc boundary.
\par As preparation for Steps 1B and 1C, the points which do not turn blue
also have to make a choice: 
\par (a) turn green if you can see at least one blue point at distance at most 
$\frac{1}{10 C_{15}}$ from yourself. Since $C_{15}$ is an 
absolute constant, this can be incorporated into the algorithm. 
\par (b) turn yellow if you see no blue point at distance less than 
$\frac{1}{10 C_{15}}$, 
but you do see a blue point at distance less than $\frac{1}{2}$,
\par (c) if you see no blue point at distance less than $\frac{1}{2}$, 
do nothing. 
\\
There is an issue of synchronisation here, but since there is an absolute
constant time in which every point can decide whether to turn blue or not, we
can also ensure that every point has decided blue/not blue before
the choices (a)-(c) are made, without needing perfectly reliable clocks. 
Assumptions 10 and 11 also 
guarantee that the latter choices
can all be made within an absolute constant time. Hence Step 1A runs in 
an absolute constant time.    
\\
\par 
To complete Step 1, the idea is that the leader should come from among the
points which are blue after Step 1A. Since there are a.a.s. $\Theta(n)$
such points, if each of them were to generate at least $C_{16} \log n$
random bits, independent of all the others, then for $C_{16}
\gg 0$, a.a.s. there would be a unique point which generates the
largest binary number. There are basically three problems to be solved in
order to turn this into an algorithm:
\par (i) how do the blue points know how many bits to generate ?
\par (ii) how do they store the bits ?
\par (iii) how do they check their numbers against those generated by 
all other blue points ?
\\
Step 1B will deal with (i) and (ii),
whereas Step 1C will deal with (iii). As already mentioned, problem 
(ii) could be
ignored if we allowed each point to possess unlimited memory, and to 
overcome this is where the algorithm becomes most technical. 
\\
\\
\emph{Step 1B:} Each blue point scans all its green neighbors. 
Lemma \ref{lem:graph} implies that, for $C_4 \gg 0$, a.a.s. every 
blue point will scan at least $\frac{C_{4}}{1000C_{15}} \log n$ (say) 
and at most 
$o(n)$ green neighbors. As explained in Assumption 10, 
a blue point can actually 
scan its surroundings in 
such a way that it doesn't miss any green points and identifies them 
at a rate of $\Omega(1)$ points per
second, and hence can complete the scan in time $o(n)$. Now, since
the total number of blue points is $\lesssim C_{15} n$, and $C_{15}$ is 
an \emph{absolute} constant, 
if every blue point generated as many random bits as it had
green neighbors, then a.a.s. there would be a unique largest binary number,
provided $C_4$ is sufficiently large.  
\par Each time a green
point is scanned, the blue scanner generates a random bit. These bits now
need to be stored somewhere, and since their number a.a.s. goes to infinity,
we cannot assume the blue point stores them in its own
internal memory. This is where the yellow points enter the picture.  
\par A blue point will choose points amongst its yellow neighbors to store two 
copies of the
random number it generates, which we call the S-copy (S for ``stationary'')
and the M-copy (M for ``mobile''). The reason for two copies and for our
curious terminology will become clear in Step 1C. For present purposes, the
basic idea is that, each time a green point is scanned then a random bit
is generated and two yellow points selected to store it. Each selected
yellow point receives two bits, the first is the random bit generated by 
blue, the second determines whether it belongs to the S-copy or 
the M-copy (say 0 for S and 1 for M). As soon as a point 
accepts a storage request then it turns purple, to indicate that it is 
no longer available for such requests. Thus each storage point belongs
unambiguously to either the S- or the M-copy of a unique number. There is
one more issue to deal with, namely the bits of (a copy of) a blue chief's
number must be stored in such a way that they form an unambiguous
binary number, which can be ``read from left to right'' 
unambiguously in Step 1C. This is achieved as follows. The first time
a blue point generates a random bit and finds two yellow points to store it, 
it also signals to them that they are its \emph{headmen}. The S-headman
turns brown and the W-headman turns orange. The headmen now also
reposition themselves so that 
\par (i) each remains within distance $\frac{1}{2}$ of its blue \emph{chief}
\par (ii) the line segment from the headman to the chief is approximately at 
right angles to the boundary of the disc and directed outwards from the disc
\par (iii) the W-headman is to the right of the S-headman, seen from the
center of the disc{\footnote{Actually, (ii) and (iii) are unnecessary 
requirements, they are merely of aesthetic value. We imagine the
points of each string lining up at right angles to the
disc boundary. This ``looks nice'', but will not be needed for the 
execution of Step 1C to follow. Indeed, all that's important is that
each point along the string knows where its two neighbors are, the shape of
the string as a whole is irrelevant.}}. 
\\
Since all the points in question are a.a.s. within distance $\frac{3}{5}$ 
of the
disc boundary, they can make good estimates of the various directions involved
and, together with Assumptions 8 and 9 in particular, 
it is clear that they can reposition themselves so
as to satisfy (i)-(iii) in time $O(1)$. Their blue chief tracks their
movements (Assumption 8) and so is aware of their exact resting locations. Now 
it is time to scan the second green neighbor, generate the second random bit 
and select two yellow points to store it. As well as storing their
respective bits, these two hitherto yellow (and now purple points) now
walk toward their chief and reposition themselves somewhere on the line 
segment between their chief and the appropriate headman. Using
Assumptions 8-9, this manoeuvre can be successfully accomplished 
in time $O(1)$. Once in position, these points take over the role of headman, 
whereas the previous headmen turn new colors to indicate that they are now 
\emph{tailmen}. In subsequent steps, the selected storage points reposition
themselves somehere on the line segment between their chief and the 
current appropriate headman. Once they have done so, they take over the
role of headman, and the previous headmen turn some neutral color
to indicate that they are now \emph{middlemen}. 
\par There is one final subtlety. Since the total number of blue points is
$\lesssim C_{15} n$, on average any yellow point will be able to see
no more than $\frac{C_{15}}{\pi}$ blue points. 
Putting it another way, on average no more
than $\frac{C_{15}}{\pi}$ 
blue points will be competing for the services of any yellow 
point as storage space. Since the total number of yellow points will be close 
to $(5 C_{15})^2$ times the total number of green ones, on average there
will be enough storage space to go round. 
However, we do not see how to rule out the possibility that
there may locally be much denser clusters of blue points. A priori, it
follows from Claim \ref{clm:boundsize} that up to $O(\log n)$ blue points 
may be visible from any yellow point. In that case, it can happen that a blue 
point runs out of storage space before it has scanned all its green
neighbors. If that is the case, we shall let the blue point ``drop out of
contention'', i.e.: it abandons the generation of its random 
number. It then 
turns pink to indicate that it is out of contention for leadership, 
and this information is relayed to all the current members of its S- and
W-strings, who subsequently revert to white (we let this be the 
default color of a point that isn't doing anything). Note that, because the
points in these strings lie on a straight line from their chief, they cannot
receive signals from the latter directly, so information has to be relayed
along the string. This can still be accomplished in time $o(n)$, 
however. 
\\
\\
\emph{Step 1C:} Once a blue point has completed Step 1B, and if it has not
turned pink, it signals to the headman of its M-string to
``start walking''. The idea here is that the M-string will now 
walk around the boundary of the disc such that the interior of the disc is on 
its right{\footnote{or its left, it doesn't matter as long as each headman 
maintains
a consistent orientation. Indeed, there is no problem if different headmen 
head in opposite directions.}}. 
The headman leads the way and the other points in the 
string follow (see Assumption 8). 
As it walks, it must follow the trail of blue and pink points, and
for each blue point it must consult the S-copy of that point's number and
compare it with its own. The two numbers are compared bit-by-bit (we
specify the exact protocol below) and the comparison is aborted once
a bit is found where the two numbers differ. If the S-number has
a 1 in this position, the M-headman immediately decides that its chief
is out of contention for leadership. 
It stops walking once it returns to its starting point, which 
it recognises by the fact that the S-string it just read 
is identical to its own (it knows when it's finished reading a 
string since it recognises a tailman by the particular color reserved
for tailmen). If it
returns without having encountered any S-number 
greater than its own, it signals
to its chief that it is leader. The 
following considerations ensure that this procedure a.a.s. 
achieves what we want. 
\par Firstly, because of Claim \ref{clm:boundsize}, a.a.s. 
all the chiefs are located very close to the
disc boundary and hence a wandering headman can make sure it doesn't
miss any blue or pink points 
while keeping the interior of the disc on its right. Secondly, the 
lower bound on $f(n)$ ensures that a.a.s. exactly one chief will be declared 
leader, provided $C_4 \gg 0$. Note that it doesn't matter if some chiefs 
already dropped out
by turning pink in Step 1B, since there will a.a.s. be 
$\Omega(n)$ of them left in the game. Thirdly, by
Lemma \ref{lem:read}, a.a.s. each M-string will
complete its walk in time $O(n)$, provided it does not does not waste
any time ``queueing''. This is the only subtle issue here. Since an M-string
may read anything from one to $\Omega(\lambda_n)$ 
bits of an S-string, there may be
many M-strings reading the same S-string simoultaneously. To avoid
queues developing, we need to have an appropriate reading
protocol. One solution is to have the M-headman read the S-bits one by one: 
it is no problem for it to follow the trail of S-bits, since they 
lie in a straight line. Alternatively, each S-bit could keep a laser 
shining on its successor to point it out, see Assumption 9. 
Simoultaneously, the M-headman 
requests bits from its own string to compare with. 
Imagine that each point in the string maintains two copies of its bit, one
of which is permanent and the other in short-term memory. The headman
requests bits from the first middleman. Whenever a middleman receives a 
request for a bit, it empties the contents of its short-term memory and then
requests the contents of the short-term memory of the 
next point along the string. 
In this way, bits can be passed along the M-string to the headman at a rate of
$\Omega(1)$ bits per time unit, who can then compare them 
one-by-one with the trail of S-bits. This protocol ensures that the S-strings 
are completely ``passive'' and hence can be read simoultaneously by 
arbitrarily many M-headmen without any queues developing. 
\par One final comment: In steps 1B and 1C we have not bothered about trying
to synchronise the activities of different actors. I think it is clear that
this is not a problem. A.a.s. steps 1B and 1C will, as described, lead to one
chief identifying itself as leader in time $O(n)$. Once it does so, it can
initiate Step 2 by turning red. 
Note that, during Step 1, some points in the configuration 
will have changed their locations. However, a.a.s. 
all these points were initially 
within distance $\frac{3}{5}$ of the disc boundary and hence their movements
will not affect the connectedness of the communication graph. Hence, a.a.s.
Step 2 will subsequently still succeed and the whole algorithm 
terminate in time $O(n)$.  
\end{proof}

\setcounter{theorem}{5}

\begin{rem}\label{rem:on3}
The above algorithm would also work if $f(n) \leq \kappa n^3$ and the constant 
$\kappa$ were somehow known in advance. When $f(n) = \Theta(n^3)$, then 
the blue
points will generate $\Theta(n)$ random bits, but this is
not a problem. 
If $\kappa$ is not specified, then the difficulty is rather that the
constant we denoted $C_{15}$ will now depend on $\kappa$, since it
depends on the size of $\delta G$ via $C_8$ and $C_{11}$. In fact, 
it is only the size of $\delta_{2} G$ that is problematic, as
this will now be $\Omega_{\kappa}(n)$, whereas we can see from
(\ref{eq:expect}) that $|\delta_{1} G|$ will a.a.s. remain bounded
by an absolute constant times $n$. 
\par As explained
in the proof, the constant $C_{15}$ 
determines the amount of competition, on average, 
for yellow storage space. To be sure there
is enough storage space to go round on average, the ratio of yellow to green
points must be sufficiently large, in other words the search radius
for green points must be sufficiently small, depending on $\kappa$. We 
do not see how to get around this problem. Amusingly, though, if the points
somehow knew that their average density was $\Theta(n)$, then they would
realise that generating $\Theta(n)$ random bits is extremely 
wasteful, since it would suffice with $\Omega(\log n)$ bits. The issue
of whether there is a better alternative than counting
neighbors for the task of estimating how many random bits need to be 
generated by a blue point will be taken up in Question \ref{quest:estdiam}
below.
\end{rem}

\begin{rem}\label{rem:storage}
If we allowed each point to possess unlimited storage capacity, then the 
description of Steps 1B and 1C could be simplified considerably. There
would be no need for the colors grren and yellow. Each 
blue point could just count all its neighbors, generate one random bit
per neighbor and store the entire binary number. 
It could then choose a headman from amongst its neighbors and write 
a copy of the number to the headman's memory. The headman would then walk 
around the boundary, comparing its number with that stored in each blue point.
It would a.a.s. know it has returned to its own blue chief when it
reads a number equal to its own. 
\par The fact that our algorithm does not utilise unlimited storage capacity 
per point seems important, however. Another interesting observation is that 
the kinds of computations done in Steps 1B and 1C are
very primitive, basically only counting and comparing strings. Of course, the
sophistication of the procedure lies in the execution of
``higher-level'' tasks such as scanning, pointing and following. These
are the kinds of abilities a sceptical reader might object to. 
\end{rem}
  
\setcounter{equation}{0}


\setcounter{equation}{0}

\section{Comments and Questions}\label{sect:discuss}

In this section, whenever we formulate a precise question it should be 
understood that the 11 assumptions made in Section \ref{sect:heart} are
valid. However, just as important an issue going forward is whether there
are algorithms which work just as well under some weaker
set of assumptions about how points can interact. This should be kept in
mind at all times.
\\
\par  
For the task of merging at a single point to be strictly meaningful, 
one must assume from the outset 
that one is dealing with
point particles which are capable of locating other point particles with
infinite precision. This is true in both the classical multi-agent
Rendezvous setting (RP-5) and in ours (Assumption 4 in Section 
\ref{sect:heart}). The points are 
idealisations and each particle has, in 
``reality'' a fixed, if small size. In this paper,
we have considered ``generic'' configurations of points in a disc of
radius $n$ in $\mathbb{R}^2$. 
In order to ensure that the communication graph would be a.a.s. connected, 
it was required that the average density of particles
be $\Omega (\log n)$, hence goes to infinity with $n$. What this implies is
that, once the idealisation of point particles is removed, then
generic configurations of particles will a.a.s. not be connected and
RP-4 fails. This
seems like a serious obstacle to making practical sense of the whole
project of studying Rendezvous for generic configurations of agents. 
So, if you didn't already
consider Theorem \ref{thm:main} 
just an intellectual curiosity, then that seems even
more apparent now. However, it is not clear that all is in vain, especially
if one is satisfied with a randomised algorithm that a.a.s. succeeds, rather
than a fully deterministic procedure. 
Suppose the average 
point density is $\Theta(1)$, this obviously being the most ``realistic''
setting, from what we have just said. Almost surely, there will be isolated
points in this regime, but there will also be large connected
components, so an isolated point does not necessarily 
need to perform a Brownian motion in order to make contact with other 
points, something which (as stated earlier) would on average take an
infinite time to succeed. This leads to our first and most important question:

\begin{question}\label{quest:theta-one-uni} 
\emph{In the notation of Theorem \ref{thm:main}, 
suppose $f(n) = o(n^2 \log n)$. Does there exist a randomised
merging algorithm which a.a.s. runs in finite time and, if so, how 
quickly as a function of $n$? In particular, consider these issues when
$f(n) = \Theta (n^2)$. The same questions can be asked in the classical 
setting, i.e.: assuming RP 1-5, at least as long as we ignore the
problem of storage capacity.}
\end{question}

We suspect nevertheless that the above questions have negative answers. Hence,
perhaps a better idea is to ask the same questions, but conditioned on 
the communcation graph being connected:

\begin{question}\label{quest:theta-one-connect}
\emph{Same question as above, but conditioned on the initial
communication graph being chosen uniformly at random from among all
connected such graphs on $f(n)$ nodes. Even more ``realistically'', one
could condition on the graph being connected and the number of points
per unit area being bounded.}
\end{question}

Once one accepts that each particle has a finite size, a second obvious
conclusion is that, in reality, Rendezvous means that all the particles
come to occupy some sufficiently small region of the plane, rather than merge
at a single point{\footnote{The robustness of the ASY-algorithm to removal
of unrealistic idealisations was tested numerically in 
\cite{AOSY}. In their simulations, they replaced each point robot
by a disc of fixed area and considered $n$ robots to have
rendezvoued once they were all inside a disc of radius
$3 \sqrt{n/2}$.}}. An intermediate step, which is more
mathematically appealing, is to retain the assumption of point
particles but redefine the Rendezvous condition in this way, and see if it
leads to a significant reduction in the run-times of algorithms. In 
Theorem \ref{thm:main}, the run-time is already linear in $n$, so 
relaxing the definition of Rendezvous cannot improve performance
significantly, at least as long as we want some non-trivial convergence
of the particle swarm, that is to a disc of radius $o(n)$. We think this
emphasises the fact that our algorithm is really a procedure for
choosing a \emph{leader}, rather than for merging. Our algorithm for choosing 
a leader can in turn be broken down into two main steps. In the first step, the
graph boundary $\delta G$ is identified approximately. In the second step,
the points which think they might be in $\delta G$ assign representatives to 
walk around the disc boundary
comparing randomly generated binary strings. This second step seems
much more contrived than the first. In addition, the eventual merging will take
place at the leader's location, which is 
a.a.s. very close to the disc boundary. 
In contrast, classical merging procedures
involve some kind of repeated ``averaging'', which means that, for a generic
configuration of points in a disc, those 
closer to the boundary will gradually move inwards, and the eventual merging
will take place somewhere close to the center of the disc. 
\par This leads to at least two possible lines of further questioning. 
On the one hand, we can seek an alternative to our algorithm which more 
resembles classical merging algorithms: 

\begin{question}\label{quest:twostep} 
\emph{In the setting of Section 2, is there a linear-time 
algorithm which does not have the two-step character of ours, in which
first a leader is chosen and then all the points move to its location?}
\end{question}

On the other hand, we can seek alternative 
linear-time algorithms after relaxing
the definition of Rendezvous so that 
the points are only required to come sufficiently close together. Since 
part of the problem is to specify exactly how close, we will not 
formulate a precise question, though 
at the very least, they 
should all be required to move inside some disc of radius $o(n)$. In this
relaxation it is also natural to seek an algorithm which involves less 
sohpisticated communication between agents than that allowed by the 
assumptions of Section 
\ref{sect:heart}. Some communication is probably still necessary, however.
If we consider the ASY-algorithm, for example, and again ignore storage issues,
it was already noted in Section \ref{sect:intro} that the time taken to execute
any non-trivial convergence is
$\Omega(n \log n)$. Indeed, extending a comment made in
Section \ref{sect:intro}, we conjecture that the same is
true of any algorithm which assumes RP 1-5. It would also be interesting
to determine the actual (expected) rate of convergence for such algorithms, 
for example ASY, for
generic connected configurations. As already mentioned in Section 
\ref{sect:intro}, we are not aware of any rigorous treatment of this 
problem in the literature.     
\\
\par If, instead of questioning either the assumptions made in Section 
\ref{sect:heart} or the philosophy of considering generic
configurations, we accept these principles, then there are still
many unanswered questions. What jumps out immediately is the
upper bound on $f(n)$ assumed in Theorem \ref{thm:main}. 
This played two crucial roles:
\par (i) it implied that the size of the graph boundary was a.a.s. $O(n)$
\par (ii) it implied that each point a.a.s. had $O(n)$ neighbors, indeed
$o(n)$ neighbors, though see also Remark \ref{rem:on3}.
Thus, by counting its neighbors, a blue point had a way to estimate how many
random bits it needed to generate. 
\\
\par
First consider (i). Here it is obviously crucial that the points are 
distributed in a disc, that is, a bounded region of $\mathbb{R}^2$ with a 
``nice'' boundary. This leads us to our next two questions. The first is a bit
vague:

\begin{question}\label{quest:lebmeas} 
\emph{For which sequences $(\mathcal{D}_n)_{n=1}^{\infty}$
of Lebesgue measurable compact regions in $\mathbb{R}^2$, satisfying
Area$(\mathcal{D}_n) \rightarrow \infty$ as $n \rightarrow \infty$, 
does an analogue of Theorem \ref{thm:main} hold ? 
The analogue we have in mind here 
is that $f(n)$ points are distributed uniformly and independently in 
$\mathcal{D}_n$, where $f(n)$ grows sufficiently fast so that the 
communication graph is a.a.s. connected, but not so fast that any point is
likely to have more than $O(n)$ neighbors.}
\end{question}

An alternative track, which we find more appealing, is to
consider the multi-agent rendezvous problem on a compact 2-dimensional 
manifold without boundary. 
The simplest question, to which we have no answer, is the following:

\begin{question}\label{quest:sphere} 
\emph{Let $S_n$ denote the $2$-sphere in $\mathbb{R}^3$ 
of radius $n$. Suppose $f(n)$ points are distributed uniformly and 
independently on $S_n$, where $f(n)$ satisfies the same conditions as
in Theorem \ref{thm:main}. Under Assumptions 1-11, 
does there exist a merging algorithm which a.a.s.
runs successfully in time $O(n)$ ?}
\end{question}

We have not seen the latter
question asked even in the classical setting, where RP 1-5 are assumed.
Perhaps this is just our ignorance, or perhaps it is because, 
if one thinks of a sphere as representing the Earth, then a compass provides
a universal point of reference and rendezvous can
trivially be
accomplished in time linear in the diameter. There are two further things
to note about the problem on a $2$-sphere:
\par (a) A Brownian motion confined to $S_n$ will a.a.s. 
arrive within distance one of any point in finite time. Hence there is 
certainly \emph{some} merging procedure on $S_n$ which a.a.s.
runs in finite time. For example, two points could merge once they see each
other, with a point choosing a neighbor at random to merge with if it
has more than one. Once points become isolated they could perform a 
Brownian walk. The expected running time of even this rather stupid
algorithm is unclear to us, however, especially in light of 
Assumption 1 in Section \ref{sect:heart}. 
\par (b) Something close to Question \ref{quest:sphere} arises if we try to 
extend our algorithm to higher dimensions. Consider an initial configuration
of points inside a ball in $\mathbb{R}^3$. Step 1A of our
procedure would identify a subset of blue points close to the
boundary of the ball, hence close to a $2$-sphere. We are not quite left
with Question \ref{quest:sphere}, since these points can also 
move through the interior of the ball now. However, these 
observations suggest that a different approach may also be necesary in
higher dimensions. 
\\
\par
Now suppose we turn instead to point (ii). There is another
obvious question:

\begin{question}\label{quest:denser} 
\emph{Can Theorem \ref{thm:main} be extended to even denser 
configurations of points ? In other words, is there a linear-time
merging algorithm on the disc which succeeds even when the average
point density is super-linear ?}
\end{question}

An additional observation here is that the method for generating random numbers
described in Step 1B of our algorithm is very inefficient when $f(n) \gg n^2
\log n$ since, as long as $f(n)$ grows polynomially in $n$, 
it would suffice for each blue point to generate
$\Omega(\log n)$ bits, whereas in fact each blue point
generates on the order of $f(n)/n^2$ bits. An alternative procedure we
considered is to have the blue points estimate $n$, the diameter of the disc.
We considered a specific, and rather complicated, signalling procedure, 
whereby the blue points send signals to their neighbors
which are relayed inwards in such a manner that, a.a.s., the signals bounce back
only once they have reached very close to the center of the disc. The details
are not important as we could not prove that our procedure worked, but
we can state another question:

\begin{question}\label{quest:estdiam} 
\emph{As an alternative to Step 1B, is there a 
procedure by which the blue points can 
estimate the diameter $n$ of the disc, up to an absolute multiplicative factor
say, a.a.s. in time $O(n)$ ?}
\end{question}

We conclude with a couple more questions which seem interesting, though
are perhaps less important. 
One striking aspect of our ``choose a leader then move to it'' algorithm is that
all the points merge at one place. This is also true
of classical rendezvous procedures. Noting also Assumption 1, we can ask

\begin{question}\label{quest:intermediate} 
\emph{In the setting of Section 2, is there a linear-time
algorithm which includes merging of groups of points at intermediate 
steps and at different locations, 
rather than all having all the points merge only at the end at
a single location ?}
\end{question}

We pose our last question as it may be interesting in its own right,
as a problem in computational geometry:
 
\begin{question}\label{quest:exactbdry} 
\emph{In the setting of Section 2, is there a 
linear-time algorithm for identifying the graph boundary 
exactly?}
\end{question}

Beyond the discussion above, several other possible lines of investigation 
present themselves 
naturally, though they take us into even less well-defined 
territory. For example, one might consider initial configurations of 
points which are not uniform i.i.d. Ultimately, one can question the 
various assumptions listed in Section \ref{sect:heart}, and try to formulate an 
interesting problem under some other set of assumptions. Last but not least,
in the setting of Section \ref{sect:heart}, 
is there a really simple linear-time algorithm
which we have completely missed, which would render most of this
manuscript unnecessary ? 






\end{document}